\documentclass[conference]{IEEEtran}
\IEEEoverridecommandlockouts

\usepackage{cite}
\usepackage{amsmath,amssymb,amsfonts}
\usepackage{bm}           
\usepackage{graphicx}
\usepackage{textcomp}
\usepackage{xcolor}
\usepackage{booktabs}      
\usepackage{url}           
\usepackage[caption=false,font=footnotesize]{subfig}  
\usepackage{placeins}      

\DeclareMathOperator{\tr}{tr}
\DeclareMathOperator{\diag}{diag}
\DeclareMathOperator{\Reop}{Re}

\newcommand{\CN}{\mathcal{CN}}

\newtheorem{proposition}{Proposition}
\newtheorem{corollary}{Corollary}
\newtheorem{definition}{Definition}
\newtheorem{remark}{Remark}

\def\BibTeX{{\rm B\kern-.05em{\sc i\kern-.025em b}\kern-.08em
    T\kern-.1667em\lower.7ex\hbox{E}\kern-.125emX}}

\begin{document}

\title{Wideband Compressed-Domain Cram\'{e}r--Rao Bounds for
Near-Field XL-MIMO: Data and Geometric Diversity Decomposition}

\author{\IEEEauthorblockN{R{\i}fat Volkan \c{S}enyuva,~\IEEEmembership{Member,~IEEE}%
\thanks{Code and data: \protect\url{https://github.com/rvsenyuva/wb-nf-crb-globecom26}.
Concept DOI: \protect\url{https://doi.org/10.5281/zenodo.19487207} (all versions).}
\thanks{\copyright{} 2026 IEEE. Personal use of this material is permitted. Permission from IEEE must be obtained for all other uses, in any current or future media, including reprinting/republishing this material for advertising or promotional purposes, creating new collective works, for resale or redistribution to servers or lists, or reuse of any copyrighted component of this work in other works. Accepted for publication in GLOBECOM 2026 -- 2026 IEEE Global Communications Conference.}}
\IEEEauthorblockA{\textit{Department of Electrical-Electronics Engineering}\\
\textit{Maltepe University}\\
Istanbul, Turkey\\
rifatvolkansenyuva@maltepe.edu.tr}
}

\maketitle

\begin{abstract}
Wideband orthogonal frequency-division multiplexing (OFDM) over
near-field extremely large-scale MIMO (XL-MIMO) arrays couples beam
squint and wavefront curvature, and existing single-frequency
compressed covariance models are severely biased as a result.
To the best of our knowledge, no compressed-domain Cram\'{e}r--Rao
bound (CRB) has been reported for this regime under hybrid
analog--digital architectures; existing wideband near-field bounds
assume full-array observation.
We derive the wideband compressed-domain CRB and decompose its
Fisher information gain into two terms: a dominant data-diversity
term, exactly $10\log_{10} K_s$ dB, where $K_s$ denotes the number
of independent subcarrier observations, and we prove that for a
centred array the full-array range bound is invariant to beam
squint, so the full-array wideband gain equals the data-diversity
term to within $3\times10^{-4}$ dB.
Under hybrid compression, however, this invariance breaks: the
residual is a combiner-dependent fluctuation of either sign rather
than a propagation-geometry effect, with mean $+0.157$ dB and
standard deviation $0.767$ dB over 385 realizations at
$N_\mathrm{RF} = 16$, $B = 400$ MHz, and its spread shrinks toward
zero as the number of RF chains grows.
At 28 GHz with $B = 400$ MHz, data diversity contributes $+27.1$ dB
and hybrid compression contributes an additional $12.6$ dB gap
relative to the full-array bound at $N_\mathrm{RF} = 16$ RF chains.
Frequency-aware covariance modeling is therefore the dominant
requirement, and the residual is a property of the analog front
end rather than of the propagation geometry.
\end{abstract}


\section{Introduction}
\label{sec:intro}

Extremely large-scale MIMO (XL-MIMO) arrays with hundreds or
thousands of antennas are a cornerstone of sixth-generation (6G)
wireless networks~\cite{liu2023nftutorial}.
At millimetre-wave (mmWave) and sub-THz carrier frequencies, the
array aperture becomes electrically large relative to the signal
bandwidth, so that two propagation effects coexist:
(i)~\emph{near-field} spherical-wave (Fresnel) propagation,
under which each scatterer is characterized by both an angle and
a range; and
(ii)~\emph{wideband beam squint}, under which the effective
electrical spacing of the array varies across OFDM
subcarriers~\cite{wang2024tcom}.
When both effects act simultaneously, the per-subcarrier array
response acquires a frequency-dependent quadratic phase that
couples angle, range, and frequency --- a phenomenon we term the
\emph{beam-squint~$\times$~Fresnel interaction} (frequency-dependent
beam squint modulating the Fresnel curvature term).

Existing wideband near-field channel estimators are predominantly
based on polar-domain compressed sensing (CS).
Cui and Dai~\cite{cui2023bpd} introduced bilinear pattern
detection (BPD) to exploit the linear frequency dependence of
the polar-domain support, building on their polar-domain
framework~\cite{cui2022polar}.
These methods, however, assume full-array digital access and do
not exploit the statistical structure of the compressed covariance
under hybrid analog--digital architectures~\cite{venugopal2017,heath2016overview}.
Meanwhile, covariance-domain estimation via
Kullback--Leibler (KL) divergence
fitting~\cite{stoica1990perf,ottersten1998cov,pote2023tsp}
has been applied to near-field channels only in the narrowband
regime~\cite{senyuva2026clkl}, to wideband channels only in
the far field~\cite{park2024sam}, and to hybrid-compressed
far-field arrays only in the narrowband
regime~\cite{senyuva2026sensors}.
The intersection of all four elements --- structured covariance
fitting, near-field Fresnel geometry, wideband OFDM, and hybrid
compression --- remains unaddressed by this prior line of work.
An efficient covariance-domain estimator closing exactly this gap
is developed separately in~\cite{senyuva2026wbclkl}; the present
work instead derives the compressed-domain performance bound for
this same setting.

On the performance-bounds side, wideband near-field CRBs
have recently been derived by
Wei et~al.~\cite{wei2025wbcrb} for position, velocity, and radar
cross-section, and by Wang et~al.~\cite{wang2025twc} for angle,
range, and complex gain.
Both works assume full-array observation
($\mathbf{W} = \mathbf{I}_M$) and do not account for the
information loss introduced by hybrid compression.
To the best of our knowledge, a compressed-domain wideband
CRB for hybrid near-field channel estimation has not
previously been reported.
A complementary line of work by
Wymeersch~\cite{wymeersch2020icc} analyzed the Fisher
information of joint near-field localization and clock
synchronization under wideband OFDM signaling with a
full uncompressed array, and concluded that
curvature-induced wavefront information dominates the
spatial-wideband contribution in that regime.
Wei et~al.~\cite{wei2025wbcrb} report a related finding in
their full-array wideband OFDM location bound: as array
aperture grows, near-field curvature and element-dependent
range variations dominate the position sensitivity that the
far-field approximation fails to capture.
The present analysis reaches a complementary conclusion in
the OFDM covariance domain under hybrid compression: the
per-subcarrier data-diversity term dominates the wideband
CRB gain, while the residual observed under hybrid
compression is a secondary, combiner-dependent effect at
current 5G New Radio (NR) fractional bandwidths.

This paper addresses the gap.
Our contributions are threefold.
First, we derive the per-subcarrier compressed covariance under
the wideband Fresnel--OFDM model and show that the frequency
ratio $\alpha_k = f_k/f_c$ jointly scales the linear (squint)
and quadratic (curvature) phase terms, producing a covariance
mismatch --- the Frobenius-norm deviation of the per-subcarrier
covariance from its center-frequency value --- up to 64\% at
$B = 100$~MHz and 177\% at $B = 400$~MHz.
Second, we derive the wideband compressed-domain CRB by summing
per-subcarrier Slepian--Bangs Fisher information matrices (FIMs)
and decompose the information gain into a dominant data-diversity
component ($\approx 10\log_{10}K_s$~dB, $K_s$ the number of
subcarriers used in the wideband FIM) that is exact for the full
array, and a secondary residual observed under hybrid compression.
Third, we compare the proposed compressed CRB with the full-array
bounds of~\cite{wei2025wbcrb,wang2025twc} and quantify the
compression loss as a function of the number of RF chains
$N_\mathrm{RF}$.

The remainder of the paper develops the wideband near-field
system model (Section~\ref{sec:model}), the compressed-domain
CRB (Section~\ref{sec:crb}) and its
range-gain decomposition (Section~\ref{sec:decomp}), reports numerical
results (Section~\ref{sec:results}), and concludes
(Section~\ref{sec:conclusion}).

\section{System Model}
\label{sec:model}

We consider a base station (BS) equipped with an $M$-element
uniform linear array (ULA) of half-wavelength spacing
$d_\mathrm{ant} = \lambda_c/2$, where $\lambda_c = c/f_c$ is the
carrier wavelength.
The BS communicates over an OFDM waveform with $K$~subcarriers
spaced by $\Delta f$ (5G~NR numerology~3: $\Delta f = 120$~kHz),
yielding bandwidth $B = K \Delta f$.
A hybrid analog--digital architecture with $N_\mathrm{RF} \ll M$
radio-frequency (RF) chains compresses the received signal before
digital processing.

\subsection{Fresnel Steering Vector}
\label{sec:fresnel}

Let $\bar{m} = m - (M\!-\!1)/2$ denote the centered element index.
Under the Fresnel (uniform spherical-wave) approximation, the
phase at the $m$th element due to a source at angle~$\theta$ and
range~$r$ consists of a linear term (angle) and a quadratic term
(curvature)~\cite{grosicki2005wlp}.
We parameterize these as
\begin{equation}
\omega(\theta) = -\frac{2\pi d_\mathrm{ant}}{\lambda_c}\cos\theta,
\qquad
\kappa(\theta,r) = \frac{\pi d_\mathrm{ant}^2}{\lambda_c}
  \frac{\sin^2\!\theta}{r},
\label{eq:omega_kappa}
\end{equation}
Since $d_\mathrm{ant} = \lambda_c/2$, this reduces to
$\omega(\theta) = -\pi\cos\theta$ for the angular term, while the
curvature term retains an explicit $\lambda_c$ dependence,
$\kappa(\theta,r) = (\pi\lambda_c/4)\sin^2\!\theta/r$.
The $m$th entry of the narrowband near-field steering vector is
then
$[\mathbf{a}(\theta,r)]_m =
\exp(j\omega\bar{m} - j\kappa\bar{m}^2)$.
We adopt the phase-only Fresnel/USW model standard in prior
near-field CRB
literature~\cite{grosicki2005wlp,wei2025wbcrb,wang2025twc}:
element-dependent amplitude variations of order $D_\mathrm{ap}/r$
are a small correction dominated by the phase information at the
ranges considered here.
At the Section~\ref{sec:results} default range $r = 5$~m the
aperture-to-range ratio is $D_\mathrm{ap}/r \approx 0.27$; it
exceeds unity at the shortest swept range $r = 1$~m, where the
approximation is correspondingly looser.

\subsection{Wideband Frequency Scaling}
\label{sec:ofdm}

The subcarrier frequencies are
$f_k = f_c + \big(k - (K+1)/2\big)\Delta f$ for $k = 1,\ldots,K$.
Define the frequency ratio
\begin{equation}
\alpha_k \triangleq \frac{f_k}{f_c}
  = 1 + \frac{\big(k - (K+1)/2\big)\Delta f}{f_c}.
\label{eq:alpha_k}
\end{equation}
Because the physical element spacing is fixed at
$d_\mathrm{ant}$, the effective electrical spacing at the $k$th
subcarrier is $d_\mathrm{ant} f_k / f_c = \alpha_k d_\mathrm{ant}$.
The factor~$\alpha_k$ therefore scales \emph{both} the linear phase
(beam squint) and the quadratic phase (curvature) simultaneously.
For the $\ell$th propagation path, the frequency-scaled Fresnel
steering vector at subcarrier~$k$ is
\begin{equation}
[\mathbf{a}_{\ell,k}]_m
  = \exp\!\big(j\,\alpha_k\,\omega_\ell\,\bar{m}
             - j\,\alpha_k\,\kappa_\ell\,\bar{m}^2\big),
\label{eq:a_k}
\end{equation}
where $\omega_\ell \triangleq \omega(\theta_\ell)$ and
$\kappa_\ell \triangleq \kappa(\theta_\ell, r_\ell)$.
When $\alpha_k = 1$ (center frequency), \eqref{eq:a_k} reduces to
the narrowband near-field steering vector.
When $\kappa_\ell = 0$ (far field), only the linear beam-squint
effect remains.

\subsection{Hybrid Compression}
\label{sec:hybrid}

The analog combiner $\mathbf{W} \in \mathbb{C}^{M \times N_\mathrm{RF}}$
is applied in the RF domain \emph{before} the OFDM FFT and is
therefore frequency-flat~\cite{venugopal2017,heath2016overview}.
Each entry satisfies the constant-modulus constraint
$|[\mathbf{W}]_{m,n}| = 1/\sqrt{M}$.
The compressed observation at subcarrier~$k$ and snapshot~$n$ is
\begin{equation}
\mathbf{y}_k(n) = \mathbf{W}^H \mathbf{x}_k(n)
  \in \mathbb{C}^{N_\mathrm{RF}},
\label{eq:y_k}
\end{equation}
where the full-array signal is
$\mathbf{x}_k(n) = \sum_{\ell=1}^{d} s_{\ell,k}(n)\,\mathbf{a}_{\ell,k}
 + \mathbf{w}_k(n)$
with i.i.d.\ path gains
$s_{\ell,k}(n) \sim \CN(0, p_\ell)$, independent across both
$\ell$ and~$k$, and noise
$\mathbf{w}_k(n) \sim \CN(\mathbf{0}, N_0\mathbf{I}_M)$.
The path power $p_\ell$ carries no subcarrier index,
corresponding to a fixed per-subcarrier transmit power
(equivalently, fixed per-subcarrier SNR) convention adopted
throughout; Remark~\ref{rem:energy_norm} discusses the
alternative fixed-total-pilot-energy convention.
Independence across~$k$ is the standard pilot-design
assumption: distinct pilot symbols are transmitted on
different OFDM subcarriers, so the per-subcarrier
compressed snapshots $\{\mathbf{y}_k(n)\}_{k=1}^{K}$ are
mutually independent and the joint log-likelihood factorises
over~$k$, consistent with the wideband-MIMO
modeling convention adopted
in~\cite{wei2025wbcrb,wang2025twc}.
In a physical wideband channel the same path carries a
deterministic delay-dependent phase evolution across
subcarriers; under orthogonal pilot probing this evolution
is absorbed into the per-subcarrier steering vector via the
$\alpha_k$-scaling of~\eqref{eq:a_k} and does not by itself
induce statistical dependence between the pilot observations
$\{\mathbf{y}_k(n)\}$.
This is a statistical observation-model assumption for
covariance-domain CRB evaluation under orthogonal pilot
probing, not a claim of physically independent propagation
mechanisms across frequency.

\subsection{Per-Subcarrier Compressed Covariance}
\label{sec:cov_k}

Under the stochastic signal model, the compressed covariance at
subcarrier~$k$ is
\begin{equation}
\mathbf{R}_{y,k}(\boldsymbol{\eta})
  = \sum_{\ell=1}^{d} p_\ell\,
    \mathbf{d}_{\ell,k}\,\mathbf{d}_{\ell,k}^H
  + N_0\,\mathbf{W}^H\mathbf{W},
\label{eq:Ry_k}
\end{equation}
where $\mathbf{d}_{\ell,k} \triangleq \mathbf{W}^H\mathbf{a}_{\ell,k}
\in \mathbb{C}^{N_\mathrm{RF}}$ is the compressed steering vector
and
$\boldsymbol{\eta} = [\omega_1,\ldots,\omega_d,\,
\kappa_1,\ldots,\kappa_d,\,p_1,\ldots,p_d,\,N_0]^T
\in \mathbb{R}^{3d+1}$
collects all unknown parameters.
The parameters are \emph{shared} across subcarriers; only the
$\alpha_k$-scaling of the steering vector changes with~$k$.
The sample covariance is estimated from $N$~snapshots as
$\widehat{\mathbf{R}}_{y,k} = \frac{1}{N}\sum_{n=1}^{N}
\mathbf{y}_k(n)\mathbf{y}_k(n)^H$.
In practice, only $K_s \le K$ uniformly-spaced subcarriers
are needed for CRB evaluation (Section~\ref{sec:crb}),
since the per-subcarrier FIM varies smoothly with~$\alpha_k$.
Specifically, we set $K_s = \min(K, K_s^{\max})$ with
$K_s^{\max} = 512$: for narrow bandwidths ($K < 512$) all
available subcarriers are used, while for wider bandwidths
($K > 512$) the cap reduces the observation count, which lowers the
total Fisher information by $10\log_{10}(K/K_s)$~dB while leaving
the geometric residual $\Delta_\mathrm{GD}$ of
Proposition~\ref{prop:geom_div} unchanged to within $0.006$~dB.
This convention is used throughout
Section~\ref{sec:results}.

\begin{remark}[Covariance Mismatch]
\label{rem:mismatch}
At the center frequency ($\alpha_{k_c} = 1$), the compressed
covariance reduces to the narrowband model
of~\cite{senyuva2026clkl}.
At edge subcarriers, however, $\alpha_k$ deviates from unity by
up to $\pm B/(2f_c)$.
The Frobenius-norm mismatch
$\|\mathbf{R}_{y,k} - \mathbf{R}_{y,k_c}\|_F /
 \|\mathbf{R}_{y,k_c}\|_F$
reaches 64\% at $B = 100$~MHz, 177\% at $B = 400$~MHz,
and 194\% at $B = 800$~MHz
($f_c = 28$~GHz, $M = 256$, $r \in [1, 100]$~m).
An estimator that ignores this frequency dependence suffers
model-mismatch bias proportional to bandwidth, motivating the
wideband treatment in Sections~\ref{sec:crb}--\ref{sec:decomp}.
\end{remark}

\section{Wideband Compressed-Domain CRB}
\label{sec:crb}

We derive the Cram\'{e}r--Rao bound (CRB) for the
wideband compressed observation model introduced in
Section~\ref{sec:model}.
By the pilot-design assumption stated below~\eqref{eq:y_k},
the per-subcarrier compressed snapshots
$\mathbf{y}_k(n) \sim \CN(\mathbf{0}, \mathbf{R}_{y,k})$
are mutually independent across~$k$ for each snapshot
index~$n$, so the joint log-likelihood factorises
and the wideband FIM decomposes as a sum of per-subcarrier
contributions~\cite{wei2025wbcrb,stoica1990perf}.

\subsection{Per-Subcarrier Slepian--Bangs FIM}
\label{sec:crb_fim}

Under the stochastic (unconditional) signal model, the
negative log-likelihood at the $k$th subcarrier (normalized
by~$N$) is~\cite{stoica1990perf}
\begin{equation}
\mathcal{L}_k
  = \log\det\mathbf{R}_{y,k}
  + \tr\!\big(\mathbf{R}_{y,k}^{-1}\widehat{\mathbf{R}}_{y,k}\big),
\label{eq:kl_k}
\end{equation}
which is the KL divergence between the sample covariance
$\widehat{\mathbf{R}}_{y,k}$ and the model covariance
$\mathbf{R}_{y,k}(\boldsymbol{\eta})$.

Recall the shared parameter vector from Section~\ref{sec:model}:
\begin{equation}
\boldsymbol{\eta}
  = \big[\omega_1,\ldots,\omega_d,\;
         \kappa_1,\ldots,\kappa_d,\;
         p_1,\ldots,p_d,\;
         N_0\big]^T
  \!\in \mathbb{R}^{3d+1},
\label{eq:eta}
\end{equation}
where $\omega_\ell = -(2\pi d_\mathrm{ant}/\lambda_c)\cos\theta_\ell$
is the spatial frequency and
$\kappa_\ell = (\pi d_\mathrm{ant}^2/\lambda_c)\sin^2\!\theta_\ell / r_\ell$
is the Fresnel curvature of the $\ell$th path.
The per-subcarrier Slepian--Bangs FIM is~\cite{stoica1990perf}
\begin{equation}
[\mathbf{J}_k]_{ij}
  = N \cdot \Reop\!\Big\{
      \tr\!\big(
        \mathbf{R}_{y,k}^{-1}
        \tfrac{\partial \mathbf{R}_{y,k}}{\partial \eta_i}\,
        \mathbf{R}_{y,k}^{-1}
        \tfrac{\partial \mathbf{R}_{y,k}}{\partial \eta_j}
      \big)
    \Big\},
\label{eq:fim_k}
\end{equation}
with $\mathbf{J}_k \in \mathbb{R}^{(3d+1)\times(3d+1)}$
for each $k \in \{1,\ldots,K\}$.

\subsection{Steering Vector Derivatives with $\alpha_k$ Scaling}
\label{sec:crb_deriv}

At the $k$th subcarrier, the frequency-scaled Fresnel
steering vector is (cf.\ Section~\ref{sec:model})
\begin{equation}
[\mathbf{a}_{\ell,k}]_m
  = \exp\!\big(
      j\,\alpha_k\omega_\ell\,\bar{m}
    - j\,\alpha_k\kappa_\ell\,\bar{m}^2
    \big),
\label{eq:a_lk}
\end{equation}
where $\alpha_k = f_k/f_c$ is the frequency ratio and
$\bar{m} = m - (M\!-\!1)/2$ is the centered element index.
The derivatives with respect to the spatial-frequency
and curvature parameters are
\begin{equation}
\frac{\partial \mathbf{a}_{\ell,k}}{\partial \omega_\ell}
  = j\,\alpha_k\,\bar{\mathbf{m}} \odot \mathbf{a}_{\ell,k},
\quad
\frac{\partial \mathbf{a}_{\ell,k}}{\partial \kappa_\ell}
  = -j\,\alpha_k\,\bar{\mathbf{m}}^{\odot 2} \odot \mathbf{a}_{\ell,k},
\label{eq:da_derivs}
\end{equation}
with $\bar{\mathbf{m}} = [\bar{m}_0,\ldots,\bar{m}_{M-1}]^T$.
The factor~$\alpha_k$ multiplying both derivatives is the
key structural difference from the narrowband CRB
in~\cite{senyuva2026clkl}: it causes each subcarrier to
``see'' the array at a different effective electrical length,
producing frequency-dependent Fisher information.

Define the compressed steering vector
$\mathbf{d}_{\ell,k} \triangleq \mathbf{W}^H \mathbf{a}_{\ell,k}
\in \mathbb{C}^{N_\mathrm{RF}}$.
The covariance derivatives needed in~\eqref{eq:fim_k} are
\begin{align}
\frac{\partial \mathbf{R}_{y,k}}{\partial \omega_\ell}
  &= p_\ell \!\bigg(
       \mathbf{W}^H
       \frac{\partial \mathbf{a}_{\ell,k}}{\partial \omega_\ell}
       \mathbf{d}_{\ell,k}^H
     + \mathbf{d}_{\ell,k}
       \Big(\mathbf{W}^H
       \frac{\partial \mathbf{a}_{\ell,k}}{\partial \omega_\ell}
       \Big)^{\!H}
     \bigg),
\label{eq:dRy_domega}\\[4pt]
\frac{\partial \mathbf{R}_{y,k}}{\partial \kappa_\ell}
  &= p_\ell \!\bigg(
       \mathbf{W}^H
       \frac{\partial \mathbf{a}_{\ell,k}}{\partial \kappa_\ell}
       \mathbf{d}_{\ell,k}^H
     + \mathbf{d}_{\ell,k}
       \Big(\mathbf{W}^H
       \frac{\partial \mathbf{a}_{\ell,k}}{\partial \kappa_\ell}
       \Big)^{\!H}
     \bigg),
\label{eq:dRy_dkappa}\\[4pt]
\frac{\partial \mathbf{R}_{y,k}}{\partial p_\ell}
  &= \mathbf{d}_{\ell,k}\,\mathbf{d}_{\ell,k}^H,
     \qquad
     \frac{\partial \mathbf{R}_{y,k}}{\partial N_0}
     = \mathbf{W}^H \mathbf{W}.
\label{eq:dRy_dN0}
\end{align}
Equations \eqref{eq:dRy_domega}--\eqref{eq:dRy_dN0} reduce
to the narrowband expressions in~\cite{senyuva2026clkl}
when $K=1$ and $\alpha_k = 1$.

\subsection{Wideband FIM and CRB}
\label{sec:crb_wb}

Because the subcarrier observations are mutually independent,
the wideband FIM over $K_s$ selected subcarriers is
\begin{equation}
\mathbf{J}_\mathrm{WB}
  = \sum_{k=1}^{K_s} \mathbf{J}_k
  \in \mathbb{R}^{(3d+1)\times(3d+1)}.
\label{eq:fim_wb}
\end{equation}
The dimension of $\mathbf{J}_\mathrm{WB}$ is determined solely by
the number of unknown parameters~$(3d+1)$ and does not grow
with~$K_s$.
Each additional subcarrier contributes a positive-semidefinite
term~$\mathbf{J}_k \succeq \mathbf{0}$, so the wideband FIM
is at least as large (in the L\"{o}wner sense) as any
single-subcarrier FIM: $\mathbf{J}_\mathrm{WB} \succeq \mathbf{J}_k$
for all~$k$.

\paragraph{SVD pseudoinverse}
The wideband CRB matrix is obtained by inverting
$\mathbf{J}_\mathrm{WB}$.
When the number of paths~$d$ is large relative to
$N_\mathrm{RF}$, the FIM can become ill-conditioned.
We therefore use the SVD pseudoinverse with tolerance
$\varepsilon_\mathrm{sv} = 10^{-6}\,\sigma_{\max}(\mathbf{J}_\mathrm{WB})$,
following~\cite{senyuva2026clkl}:
\begin{equation}
\mathbf{J}_\mathrm{WB}^{\dagger}
  = \mathbf{V}\,
    \diag\!\Big(\frac{1}{\sigma_1},\ldots,\frac{1}{\sigma_r},
                 0,\ldots,0\Big)\,
    \mathbf{V}^T,
\label{eq:fim_pinv}
\end{equation}
where $\mathbf{J}_\mathrm{WB} = \mathbf{V}\,\diag(\sigma_1,\ldots,
\sigma_{3d+1})\,\mathbf{V}^T$ is the eigendecomposition and
$r = |\{i : \sigma_i > \varepsilon_\mathrm{sv}\}|$ is the
numerical rank.

\subsection{Error Propagation to Physical Parameters}
\label{sec:crb_prop}

The CRB for the $\ell$th spatial frequency is
$[\mathbf{J}_\mathrm{WB}^{\dagger}]_{\ell\ell}$ and the CRB
for the $\ell$th curvature is
$[\mathbf{J}_\mathrm{WB}^{\dagger}]_{d+\ell,d+\ell}$.
Propagating to the physical angle~$\theta_\ell$ and
range~$r_\ell$ via
\begin{equation}
\frac{\partial\omega_\ell}{\partial\theta_\ell}
  = \frac{2\pi d_\mathrm{ant}}{\lambda_c}\sin\theta_\ell,
\qquad
\frac{\partial\kappa_\ell}{\partial r_\ell}
  = -\frac{\kappa_\ell}{r_\ell},
\label{eq:jac}
\end{equation}
the marginal CRBs for angle and range are
\begin{equation}
\mathrm{CRB}_{\theta_\ell}
  = \frac{[\mathbf{J}_\mathrm{WB}^{\dagger}]_{\ell\ell}}{%
     \big(\partial\omega_\ell/\partial\theta_\ell\big)^2},
\qquad
\mathrm{CRB}_{r_\ell}
  = \frac{[\mathbf{J}_\mathrm{WB}^{\dagger}]_{d+\ell,d+\ell}}{%
     \big(\partial\kappa_\ell/\partial r_\ell\big)^2}.
\label{eq:crb_theta_r}
\end{equation}
All CRB curves in Section~\ref{sec:results} are reported as
$\sqrt{\mathrm{CRB}_{\theta_\ell}}$ (degrees) and
$\sqrt{\mathrm{CRB}_{r_\ell}}$ (metres).
Equation~\eqref{eq:crb_theta_r} propagates the marginal curvature
variance alone, omitting the angle-dependence of the range mapping.
We use it for the scalar bounds of Section~\ref{sec:results} and for
Fig.~\ref{fig:crb_bw}; Fig.~\ref{fig:crb_range}(b) instead plots the
full-Jacobian range bound $\mathbf{g}_r^T \mathbf{C}_\eta
\mathbf{g}_r$. The conventions differ by at most $0.024$~dB at
$r = 5$~m and $0.69$~dB at $r = 1$~m.

\paragraph{Compressed vs.\ full-array CRB}
Hybrid compression discards $M - N_\mathrm{RF}$ spatial degrees of
freedom per snapshot, so the CRB
in~\eqref{eq:crb_theta_r} is strictly larger than the
full-array wideband CRBs of~\cite{wei2025wbcrb,wang2025twc}.
This compressed-domain CRB is the appropriate bound against which to
plot estimator RMSE for a hybrid architecture.

\section{Information Decomposition}
\label{sec:decomp}

The wideband FIM $\mathbf{J}_\mathrm{WB} = \sum_k \mathbf{J}_k$
aggregates Fisher information from $K_s$~subcarriers.
We decompose the resulting CRB improvement over the
narrowband (single-subcarrier) bound into two
physically distinct mechanisms: \emph{data diversity} and
what we term, for continuity with this paper's title,
\emph{geometric diversity} -- though, as
Proposition~\ref{prop:geom_div} shows, its curvature-induced
contribution is provably negligible ($\le 3\times10^{-4}$~dB). The
residual instead arises from beam squint interacting with the analog
combiner under hybrid compression, and is referred to throughout this
paper simply as the residual.

\begin{definition}[Narrowband Reference CRB]
\label{def:nb_crb}
The narrowband CRB is obtained by evaluating the
per-subcarrier FIM at the center frequency alone, i.e.,
$\mathbf{J}_\mathrm{NB} \triangleq \mathbf{J}_{k_c}$ with
$\alpha_{k_c} = 1$.
\end{definition}

\begin{definition}[Data-Diversity FIM]
\label{def:data_div}
The data-diversity FIM is the $K_s$-fold replication of the
center-frequency FIM:
$\mathbf{J}_\mathrm{DD} \triangleq K_s \cdot \mathbf{J}_\mathrm{NB}$.
This represents the information gain from having
$K_s$~independent covariance snapshots at the same frequency.
\end{definition}

\subsection{Data Diversity}

\begin{proposition}[Data Diversity]
\label{prop:data_div}
If the per-subcarrier FIMs $\mathbf{J}_k$ share the same
eigenvector structure (i.e., $\mathbf{J}_k = \beta_k\,
\mathbf{J}_\mathrm{NB}$ with scalar $\beta_k > 0$ for all~$k$),
under the independence assumption of
Section~\ref{sec:hybrid}, then
\begin{equation}
\mathbf{J}_\mathrm{WB}
  = \bigg(\sum_{k=1}^{K_s}\beta_k\bigg)\mathbf{J}_\mathrm{NB},
\label{eq:fim_scaled}
\end{equation}
and the CRB improvement over the narrowband bound is
\begin{equation}
\Delta_\mathrm{DD}
  = 10\log_{10}\!\bigg(\sum_{k=1}^{K_s}\beta_k\bigg)
  \approx 10\log_{10}(K_s)
  \;\;\text{dB},
\label{eq:gain_dd}
\end{equation}
where the approximation holds when $\beta_k \approx 1$
for all~$k$.
\end{proposition}

\begin{IEEEproof}
Under the stated condition,
$\mathbf{J}_\mathrm{WB}^{\dagger}
= (\sum_k \beta_k)^{-1}\,\mathbf{J}_\mathrm{NB}^{-1}$,
so $\mathrm{CRB}_i^\mathrm{WB} =
\mathrm{CRB}_i^\mathrm{NB}/\sum_k\beta_k$ for any~$\eta_i$.
If the independence assumption is relaxed (correlated path
gains across subcarriers), the effective number of
independent observations falls below~$K_s$ and
$\Delta_\mathrm{DD}$ shrinks accordingly, since correlated
observations contribute less than one full independent
snapshot each to the sum in~\eqref{eq:fim_scaled}.
\end{IEEEproof}

\begin{remark}[Energy-Normalization Convention]
\label{rem:energy_norm}
Proposition~\ref{prop:data_div}'s $10\log_{10}(K_s)$
data-diversity gain assumes the fixed per-subcarrier
transmit power convention of Section~\ref{sec:hybrid}.
Under a fixed-total-pilot-energy budget spread over
$K_s$ subcarriers instead, per-subcarrier SNR falls as
$1/K_s$; at the default operating point ($K_s = 512$,
$\mathrm{SNR} = 10$~dB) this implies an effective
per-subcarrier SNR of approximately $-17$~dB, near the
estimator threshold region identified in
Fig.~\ref{fig:ctA}. Rather than extrapolating from the
$\beta_k$ proportionality of
Proposition~\ref{prop:data_div}, a dedicated numerical check
re-derives the wideband FIM directly under this
renormalization. The result does not support a collapse toward
$0$~dB: the measured wideband-to-narrowband CRB change
under fixed-total-energy renormalization,
$\Delta_\mathrm{FE}$, ranges from $-0.57$~dB ($K_s = 64$)
to $-12.01$~dB ($K_s = 3333$), with
$\Delta_\mathrm{FE} = -5.01$~dB at the default operating
point -- a net loss, not a small surviving residual. This
reflects the fact that the compressed-domain Slepian--Bangs
FIM is not linear in per-subcarrier path power at the
resulting low SNR, so the $K_s$-fold observation-count
gain of Proposition~\ref{prop:data_div} is not recovered
under this convention. Proposition~\ref{prop:data_div}
itself is unaffected by this renormalization: its fixed
per-subcarrier-power convention remains the physically
natural model for OFDM pilot-based channel estimation, as
discussed in Section~\ref{sec:discussion}.
\end{remark}

\subsection{Range-Gain Decomposition}

\begin{proposition}[Full-Array Wideband Range-Gain Invariance]
\label{prop:geom_div}
Consider a centred, single-path ($d=1$) uniform linear array with
element gains either uniform or following any even-symmetric amplitude
taper, spatially white noise, per-subcarrier narrowband steering, and
the estimated parameter vector
$\boldsymbol{\eta} = [\omega, \kappa, p, N_0]^T$. Let
$\mathbf{W} = \mathbf{I}_M$ (full array). Define the common linear
phase ramp as the per-subcarrier diagonal unitary
$\boldsymbol{\Phi}_k = \mathrm{diag}\big(e^{j (\alpha_k - 1)
\omega \bar{m}}\big)$, linear in
the centred element index $\bar{m}$. Then the score for $\omega$ is
orthogonal to the scores for $(\kappa, p, N_0)$ at every subcarrier,
the Fisher block for $(\kappa, p, N_0)$ is invariant under
$\boldsymbol{\Phi}_k$, and the full-array wideband range gain is
exactly
\begin{equation}
\Delta_\mathrm{GD}(r) = 10\log_{10}(K_s) + \Delta_{\alpha^2}(B)
\;\;\text{dB},
\label{eq:gain_gd_fullarray}
\end{equation}
with
\begin{equation}
\Delta_{\alpha^2}(B) = 10\log_{10}\!\big(1 + (B/f_c)^2/12\big)
\le 3\times10^{-4}~\text{dB}
\label{eq:gd_scalar_bound}
\end{equation}
over the tested regime $B/f_c \in [0.002, 0.0286]$, and beam squint
contributes exactly zero to the full-array range gain.
\end{proposition}

\begin{corollary}[Physical-Parameter Invariance]
\label{cor:jacobian_invariance}
The map from $(\omega, \kappa)$ to $(\theta, r)$ in~\eqref{eq:jac} is a
deterministic Jacobian independent of the subcarrier phase ramp
$\boldsymbol{\Phi}_k$. The physical-parameter CRB therefore transforms
by congruence and inherits the invariance of
Proposition~\ref{prop:geom_div}: the full-array range and angle bounds
are exactly $10\log_{10}(K_s) + \Delta_{\alpha^2}(B)$ apart from their
narrowband references, with no separate correction arising from the
change of parameterisation.
\end{corollary}

\begin{remark}[Multipath]
For $d > 1$, inner products between steering vectors at distinct path
angles are oscillatory sums rather than centred moments, so the
odd/even cancellation underlying Proposition~\ref{prop:geom_div} does
not extend beyond the single-path case.
\end{remark}

\begin{IEEEproof}
Unit-modulus steering entries reduce every Fisher inner product at the
full array to a centred-aperture moment, independent of the element
phases. The curvature, power, and noise derivatives carry even powers
of the centred index $\bar{m}$, while the linear-coefficient derivative
carries an odd power; the cross entries between $\omega$ and
$(\kappa, p, N_0)$ therefore reduce to first and third moments of
$\bar{m}$, both zero under the stated symmetry. The phase ramp
$\boldsymbol{\Phi}_k$ is a diagonal unitary that commutes with the
even-power diagonal weighting in every quadratic form, so it cancels
identically and cannot alter the $(\kappa, p, N_0)$ Fisher block. This
establishes exact orthogonality and the invariance stated in
Proposition~\ref{prop:geom_div}.

The curvature derivative in~\eqref{eq:da_derivs} scales as $\alpha_k$,
so the curvature-related FIM entries scale as $\alpha_k^2$. Averaging
over a symmetric frequency band gives $\overline{\alpha^2} = 1 +
(B/f_c)^2/12$, giving the per-eigenvalue scalar gain
$\Delta_{\alpha^2}(B)$ of~\eqref{eq:gd_scalar_bound}. At the operating
fractional bandwidth $B/f_c = 1.43\times10^{-2}$,
\eqref{eq:gd_scalar_bound} evaluates to $7.4\times10^{-5}$~dB, and it
remains below $3\times10^{-4}$~dB throughout the tested regime.

Under hybrid compression ($\mathbf{W} \neq \mathbf{I}_M$), the
projector $\mathbf{P}_\mathbf{W}$ does not in general preserve the
reflection symmetry required for the cancellation above, so
$\boldsymbol{\Phi}_k$ no longer commutes exactly with the compressed
Fisher block and a residual coupling between $\omega$ and
$(\kappa, p, N_0)$ appears. Over 385 independent constant-modulus
combiners at $N_\mathrm{RF} = 16$, $B = 400$~MHz, this residual has
mean $+0.157$~dB, median $+0.091$~dB, and standard deviation
$0.767$~dB, and is positive in 54.3\% of realisations; its spread
decays with $N_\mathrm{RF}$ and vanishes as $\mathbf{W} \to
\mathbf{I}_M$, consistent with Proposition~\ref{prop:geom_div}. The
squint contribution at $\mathbf{W} = \mathbf{I}_M$ was verified
numerically to machine precision, confirming the theorem.
\end{IEEEproof}

\emph{Interpretation.}
Curvature contributes negligibly to the full-array wideband range
gain, regardless of bandwidth: that gain is exactly the
data-diversity term to within $\Delta_{\alpha^2}(B)$.
Under hybrid compression the residual is not a propagation-geometry
gain but a combiner-dependent fluctuation of either sign, whose
ensemble statistics are given in the proof above.
The analysis supports no extrapolation of this residual beyond the
tested bandwidth range.

\paragraph{Worked example (verifying the decomposition)}
\label{par:worked_example}
At the operating point
$(B, r, N_\mathrm{RF}, K_s) = (400~\mathrm{MHz},\, 5~\mathrm{m},\,
16,\, 512)$, the narrowband range bound is
$\sqrt{\mathrm{CRB}_r^\mathrm{NB}} = 11.948$~mm, the
data-diversity prediction is
$\sqrt{\mathrm{CRB}_r^\mathrm{DD}} = 528.04~\mu\mathrm{m}$
($\Delta_\mathrm{DD} = +27.093$~dB), and the true wideband
bound is $\sqrt{\mathrm{CRB}_r^\mathrm{WB}} = 487.12~\mu\mathrm{m}$.
The residual $0.701$~dB between the data-diversity prediction
and the true wideband bound is one realisation of the
compression-dependent term of
Proposition~\ref{prop:geom_div}, at the 80th percentile of the
385-draw ensemble; it is not a deterministic or representative
value.

\subsection{Comparison with Full-Array Wideband CRBs}

\begin{remark}[Relation to Prior Wideband Near-Field CRBs]
\label{rem:prior_crb}
Wei et~al.~\cite{wei2025wbcrb} derived wideband near-field CRBs for
position, velocity, and radar cross-section, and
Wang et~al.~\cite{wang2025twc} for angle, range, and complex gain,
both assuming full-array access
($\mathbf{W} = \mathbf{I}_M$).
Our bound differs in three respects:
\begin{enumerate}
\item \emph{Hybrid compression:}
  We account for the information loss through the
  $M \times N_\mathrm{RF}$ analog combiner~$\mathbf{W}$, so the
  bound follows from the
  $N_\mathrm{RF} \times N_\mathrm{RF}$ compressed covariance and is
  strictly larger than its $M \times M$ full-array counterpart.
  The gap decreases as $N_\mathrm{RF} \to M$.
\item \emph{Channel estimation parameterisation:}
  Our parameter vector
  $\boldsymbol{\eta} = [\boldsymbol{\omega}^T,
  \boldsymbol{\kappa}^T, \mathbf{p}^T, N_0]^T$
  targets channel estimation (angle, range, path powers,
  noise variance), whereas~\cite{wei2025wbcrb} parameterizes in
  terms of Cartesian position, velocity, and radar cross-section,
  and~\cite{wang2025twc} in terms of angle, range, and complex
  gain.
\item \emph{Information decomposition:}
  Propositions~\ref{prop:data_div}
  and~\ref{prop:geom_div} provide a clean separation of
  the wideband CRB gain into an exact data-diversity term
  and a full-array invariance result, which is absent
  in~\cite{wei2025wbcrb,wang2025twc}.
\end{enumerate}
\end{remark}

\section{Numerical Results}
\label{sec:results}

We evaluate the wideband compressed CRB using an $M = 256$ element ULA at
$f_c = 28$~GHz ($d_\mathrm{ant} = \lambda_c/2$), equipped with
$N_\mathrm{RF} = 16$ RF chains (default).
The OFDM waveform uses subcarrier spacing $\Delta f = 120$~kHz
over a bandwidth sweep $B \in [0.12, 800]$~MHz.
Unless otherwise stated, we assume a single propagation path ($d = 1$) at
default angle $\theta = 40^\circ$ and range $r = 5$~m, with $N = 64$
snapshots at an SNR of $10$~dB per antenna.
The single-path setting isolates the curvature and squint effects,
consistent with~\cite{wei2025wbcrb,wang2025twc}.
As a robustness check, we evaluate a two-path scene
($\theta = \{30^\circ, 50^\circ\}$, $r = \{3, 7\}$~m, equal power)
at $\mathrm{SNR} = 10$~dB and $N_\mathrm{RF} = 16$ using the fixed
seed-42 combiner. Under the tolerance in~\eqref{eq:fim_pinv}, the SVD
pseudoinverse retains four of seven singular directions, with
retained-subspace condition number $6.6 \times 10^3$. Relative to the
full inverse, the truncation changes the two range bounds by at most
$3.5 \times 10^{-4}$~dB, giving
$\sqrt{\mathrm{CRB}_{r_1}} \approx 0.40$~mm and
$\sqrt{\mathrm{CRB}_{r_2}} \approx 0.85$~mm, both finite and
well-separated from the narrowband bound.
The analog combiner $\mathbf{W}$ uses random constant-modulus entries
with fixed seed for reproducibility.

\paragraph{Fig.~\ref{fig:heatmap}: covariance mismatch}
Fig.~\ref{fig:heatmap} shows the relative Frobenius mismatch
$\delta(k,r) = \|\mathbf{R}_{y,k} - \mathbf{R}_{y,k_c}\|_F /
 \|\mathbf{R}_{y,k_c}\|_F$
as a function of frequency ratio~$\alpha_k$ (x-axis) and
range~$r$ (y-axis) for $B \in \{100, 400, 800\}$~MHz.
The mismatch exceeds the 5\% threshold (white contour) at all
three bandwidths, reaching 64\% at $B = 100$~MHz,
177\% at $B = 400$~MHz, and 194\% at $B = 800$~MHz.
This confirms that narrowband covariance models are
inadequate at wideband operation and motivates the
frequency-aware treatment of
Sections~\ref{sec:crb}--\ref{sec:decomp}.

\begin{figure}[!htbp]
\centering
\includegraphics[width=0.86\linewidth]{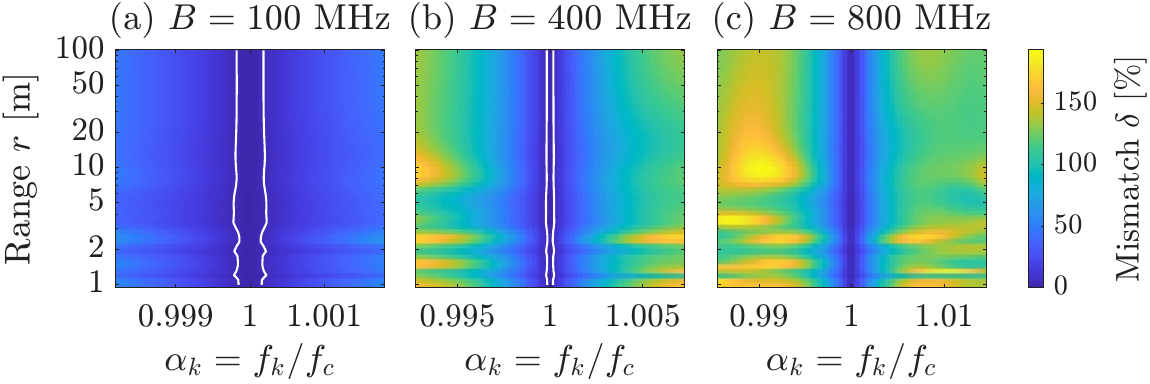}
\caption{Relative covariance mismatch $\delta(k,r)$ vs.\
frequency ratio $\alpha_k$ and range~$r$ for
$B \in \{100, 400, 800\}$~MHz.
White contour: $\delta = 5\%$.
Mismatch reaches 64\%, 177\%, and 194\% at $B = 100$, 400, 800~MHz,
motivating frequency-aware processing. At $B = 800$~MHz, the 5\% contour 
is confined to a sub-pixel strip near $\alpha_k=1$, confirming that the narrowband 
model fails across the entire operating band.}
\label{fig:heatmap}
\end{figure}

\paragraph{Fig.~\ref{fig:crb_bw}: CRB vs.\ bandwidth}
Fig.~\ref{fig:crb_bw} plots $\sqrt{\mathrm{CRB}_r}$ as a
function of bandwidth at $r = 5$~m.
Below $B \approx 60$~MHz the available subcarrier count
$K = B/\Delta f$ falls below the cap~$K_s^{\max} = 512$,
so $K_s = K$ grows with~$B$ and the CRB tracks the
$1/\sqrt{K_s}$ data-diversity scaling
predicted by Proposition~\ref{prop:data_div}.
Above $B \approx 60$~MHz, $K_s$ saturates at~$512$ and the
data-diversity contribution becomes constant; the residual
CRB decrease visible in Fig.~\ref{fig:crb_bw} for
$B \in [100, 800]$~MHz is therefore attributable to the
residual observed under hybrid compression
(Proposition~\ref{prop:geom_div}), not to a propagation-geometry
effect.
At $B = 400$~MHz ($K_s = 512$), data diversity contributes
$+27.1$~dB of the total CRB improvement over the narrowband
bound; the remainder is the combiner-dependent residual
characterized in Section~\ref{sec:decomp}.

\begin{figure}[!htbp]
\centering
\includegraphics[width=0.56\linewidth]{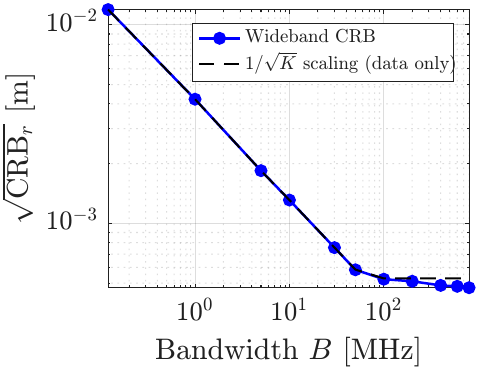}
\caption{CRB vs.\ bandwidth at $r = 5$~m, $N_\mathrm{RF} = 16$.
Dashed: $1/\sqrt{K_s}$ scaling (data diversity only).}
\label{fig:crb_bw}
\end{figure}

\paragraph{Fig.~\ref{fig:crb_nrf}: CRB vs.\ $N_\mathrm{RF}$}
Fig.~\ref{fig:crb_nrf} shows the effect of the number of RF
chains on $\sqrt{\mathrm{CRB}_r}$ at $B = 400$~MHz and
$r = 5$~m.
As $N_\mathrm{RF}$ increases from~4 to~64, the compressed CRB
decreases monotonically toward the full-array bound,
consistent with the compression loss vanishing as
$N_\mathrm{RF} \to M$.
With $N_\mathrm{RF} = 16$ (our default), the fixed-combiner gap is
$\approx 12.6$~dB; at $N_\mathrm{RF} = 32$, it narrows to
$\approx 9.4$~dB.

\begin{figure}[!htbp]
\centering
\includegraphics[width=0.56\linewidth]{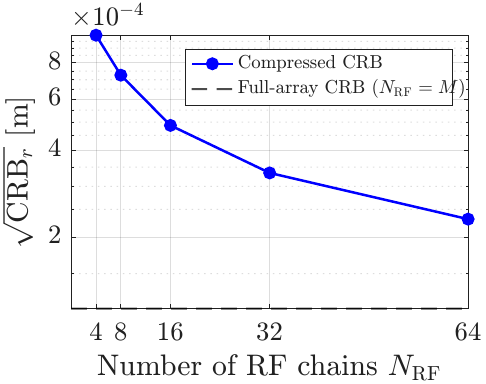}
\caption{CRB vs.\ $N_\mathrm{RF}$ at $B = 400$~MHz, $r = 5$~m.
Dashed: full-array wideband CRB.}
\label{fig:crb_nrf}
\end{figure}

\paragraph{Fig.~\ref{fig:crb_range}: CRB vs.\ range}
Fig.~\ref{fig:crb_range} plots
$\sqrt{\mathrm{CRB}_\theta}$ and $\sqrt{\mathrm{CRB}_r}$
vs.\ range at $B = 400$~MHz,
comparing the wideband compressed CRB, the narrowband
compressed CRB (single subcarrier), and the full-array
wideband CRB ($\mathbf{W} = \mathbf{I}_M$).
The vertical line marks the effective beamfocused Rayleigh
distance (EBRD), bounded by $(r_\mathrm{RD}/10)\cos^2\theta
\approx 20.4$~m at the default $\theta =
40^\circ$~\cite{hussain2025ebrd}, where $r_\mathrm{RD} =
2D_\mathrm{ap}^2/\lambda_c$ is the Rayleigh distance.
The EBRD bounds the region of finite beam depth and
polar-domain sparsity. It is shown for scale: the EBRD
governs beamfocusing geometry rather than the estimation
bound, and the range curves in panel~(b) follow smooth $r^2$ growth
throughout the tested interval.
The wideband bound is uniformly lower than the narrowband
bound by $\approx 27$~dB, and the compression gap relative to
the full-array CRB is $\approx 12.6$~dB at $r = 5$~m.

\begin{figure}[!htbp]
\centering
\includegraphics[width=0.88\linewidth]{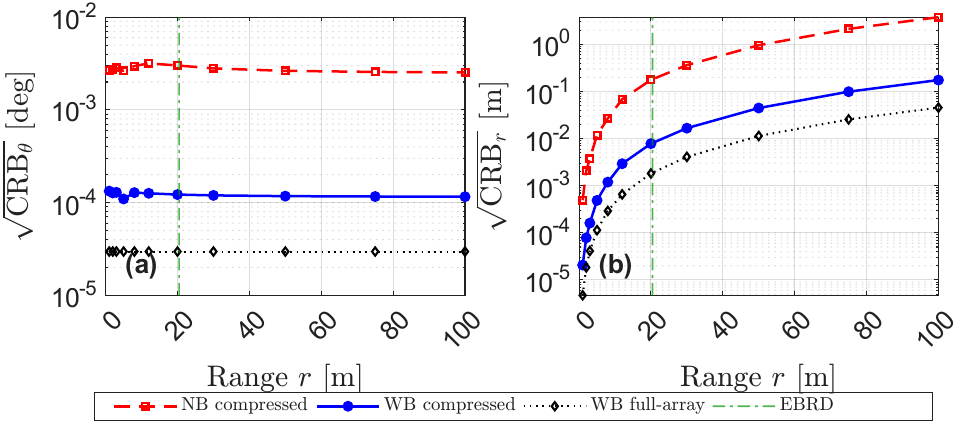}
\caption{CRB vs.\ range at $B = 400$~MHz:
(a)~$\sqrt{\mathrm{CRB}_\theta}$ [deg],
(b)~$\sqrt{\mathrm{CRB}_r}$ [m].
Dashed: narrowband compressed. Dotted: full-array wideband.
Vertical line: effective beamfocused Rayleigh distance (EBRD).
The wideband bound stays ${\approx}27$~dB below the narrowband bound.}
\label{fig:crb_range}
\end{figure}

\subsection{Synthesis and Discussion}
\label{sec:discussion}

The $+27.1$~dB data-diversity improvement reported at
$B = 400$~MHz merits careful interpretation.
It presumes the fixed per-subcarrier-power convention, under
which each subcarrier carries an independent pilot symbol and
the $K_s$-fold gain counts the additional pilot symbols
transmitted across the wideband; under the alternative
fixed-total-energy budget the gain does not persist, as
Remark~\ref{rem:energy_norm} quantifies.
Under that convention the $\sim 12.6$~dB compression gap at
$N_\mathrm{RF} = 16$ remains the dominant loss.


\subsection{Estimator Validation}
\label{sec:ctA}

The compressed-domain CRB derived in this paper is a bound on any
unbiased estimator; whether it has the correct scale and
signal-to-noise ratio (SNR) dependence is best confirmed against a
feasible estimator rather than by inspection of the FIM alone.
We construct a MUSIC-type estimator on the whitened per-subcarrier
eigenvector stack, searching jointly in $(\omega,\kappa)$ with a
sum-of-null-residuals objective across $K_s$ subcarriers; it is
included solely as a bound-scale and SNR-dependence check and is
not a proposed contribution.
Fig.~\ref{fig:ctA} reports its root-mean-square error (RMSE) against
$\sqrt{\mathrm{CRB}_r}$ over eight SNR points from $-20$ to
$10$~dB at $B = 400$~MHz, with the
default geometry and a fixed combiner.
This experiment uses $K_s = 64$ subcarriers rather than the
Section~\ref{sec:results} default of $512$.
All trials enter the reported RMSE, including boundary and
non-converged outcomes; no conditional filtering is applied.

\begin{figure}[!htbp]
\centering
\includegraphics[width=0.56\linewidth]{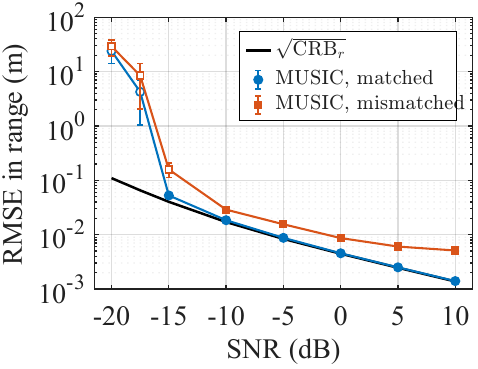}
\caption{All-trial RMSE vs.\ $\sqrt{\mathrm{CRB}_r}$ for the matched
(frequency-aware) and narrowband-mismatched MUSIC-type estimators at
$B = 400$~MHz, $\theta = 40^\circ$, $r = 5$~m, fixed combiner,
$K_s = 64$ subcarriers.
Open markers identify the threshold region: at $-20$~dB, capture is
$0.4643$ (matched) and $0.2886$ (mismatched), with gross rates $0.4400$
and $0.7186$; the large separation there reflects global
estimator breakdown and is not evidence of a loose bound.
Finite-sample scatter of $\pm 0.3$--$0.4$~dB at $N_\mathrm{MC} = 600$
is expected on the dB scale and should not be read as trend.
$N_\mathrm{MC} = 1400$ at $-20$, $-17.5$, and $-15$~dB; $600$
elsewhere.}
\label{fig:ctA}
\end{figure}

The clearest summary is the joint statistic
$\mathbb{E}[\mathbf{e}^\top \mathbf{C}^{-1}\mathbf{e}]$, where
$\mathbf{e}$ is the estimation error in $(\omega,\kappa)$ and
$\mathbf{C}$ is the CRB covariance; its expected value is exactly~$2$
however the error splits between the two coordinates, so it is
invariant to the coordinate rotation that can make either marginal
look favorable or unfavorable.
At $\mathrm{SNR} = 10$~dB the matched (frequency-aware) processor
gives a mean value of $2.019778$, a joint gap of $+0.043$~dB.
The bound has the correct scale and SNR dependence, while a finite
gap between the estimator and the bound is expected.
Below $-15$~dB the joint statistic becomes outlier-dominated by
occasional large errors, and open markers in Fig.~\ref{fig:ctA} flag
this threshold region explicitly.
The operational transition for the matched processor falls between
$-17.5$ and $-15$~dB.

A second, narrowband-mismatched processor uses the same estimator
with the per-subcarrier frequency scaling $\alpha_k$ fixed at unity.
At the tested $400$~MHz, $40^\circ$, $5$~m configuration and fixed
combiner, narrowband processing shows a $4.5$~mm range bias
over $0$--$10$~dB SNR; its RMSE loss relative to frequency-aware
processing grows
with SNR, reaching $11.3$~dB at $10$~dB SNR
($10\log_{10}(\mathrm{MSE}_\mathrm{NB}/\mathrm{MSE}_\mathrm{M})$).
The mismatched curve is biased, so its separation from the CRB is
not an efficiency gap; narrowband mismatch also worsens
threshold-region robustness.

\section{Conclusion}
\label{sec:conclusion}

We derived the wideband compressed-domain Cram\'{e}r--Rao bound
for near-field channel estimation under hybrid XL-MIMO.
First, per-subcarrier compressed covariances diverge from the
narrowband model by up to 177\% at $B = 400$~MHz, so
frequency-aware processing is necessary.
Second, the wideband FIM decomposes into a dominant data-diversity
component scaling as $10\log_{10}K_s$~dB, exact for the full array,
and a secondary residual observed under hybrid compression; data
diversity alone yields $+27.1$~dB at $B = 400$~MHz.
Third, for the fixed combiner, hybrid compression costs a further
12.6~dB relative to the
full-array CRB at $N_\mathrm{RF} = 16$, closing as $N_\mathrm{RF}$
grows; that residual is a property of the analog front end rather
than of the propagation geometry.
Companion papers develop the narrowband covariance-domain
estimation algorithm~\cite{senyuva2026clkl} and its wideband
extension~\cite{senyuva2026wbclkl}.

\end{document}